\documentclass[aps,prx,twocolumn,amsmath,amssymb,floatfix,superscriptaddress]{revtex4-2} 
\usepackage{amsmath,amssymb,amsfonts,graphicx,bm,array,physics}
\usepackage{natbib}

\usepackage{bbm}
\usepackage[outline]{contour}
\usepackage{color}
\usepackage{dcolumn}
\usepackage[utf8]{inputenc}
\usepackage{enumitem}
\usepackage{epstopdf}
\usepackage{amsthm}
\usepackage{empheq}
\usepackage{braket}
\usepackage{mathtools}
\usepackage[usenames,dvipsnames]{xcolor}
\usepackage{latexsym}
\usepackage{hyperref} 
\hypersetup{colorlinks = true,allcolors = blue }
\usepackage[capitalize]{cleveref}
\usepackage{diagbox} 
\usepackage[caption=false]{subfig}
\usepackage{multirow}
\usepackage{outlines}

\usepackage{epsf}

\usepackage{xspace}

\newcommand{\LP}[1]{\textcolor{green!50!black}{#1}}

\newtheorem{theorem}{Theorem}
\newtheorem{corollary}{Corollary}

\def\Sep{ \textnormal{Sep} }

\def\Incoh{ \textnormal{Incoh} }
\def\Coh{ \textnormal{Coh} }
\def\Mag{ \textnormal{Mag} }
\def\Stab{ \textnormal{Stab} }

\def\Rob{ \mathcal{R}_\mathcal{L} }

\def\Free{ \mathcal{L} }
\def\out{ \textnormal{out} }
\def\cost{ \textnormal{cost} }

\newcommand{\id}{\mathbbm{I}}

\newcommand{\bseq}{\begin{subequations}}
\newcommand{\eseq}{\end{subequations}}
\newcommand{\bsplit}{\begin{split}}
\newcommand{\esplit}{\end{split}}

\renewcommand{\tr}[1]{\operatorname{\textnormal{Tr}}\left( {#1} \right)}

\begin{document}

\newcommand{\lanl}{Quantum and Condensed Matter Physics Group (T-4), Theoretical Division, Los Alamos National Laboratory, Los Alamos, New Mexico 87545, USA}
\newcommand{\iiserk}{Department of Physical Sciences, Indian Institute of Science Education and Research Kolkata, Mohanpur, West Bengal 741246, India}
\author{Luis Pedro Garc\'ia-Pintos}
\email[Corresponding author: ]{lpgp@lanl.gov}
\affiliation{\lanl}
\author{Tanmoy Biswas}
\email{tanmoy.biswas23@gmail.com}
\affiliation{\lanl}
\author{Chandan Datta}
\email{chandan@iiserkol.ac.in}
\affiliation{\iiserk}

\title{
Robustness as a thermodynamic currency: \\ work advantages and preparation costs of nonclassical states
}

\date{\today}
\begin{abstract}
Understanding whether uniquely quantum features can provide concrete advantages in thermodynamic processes is a central objective of quantum thermodynamics. A key challenge is quantifying how different forms of non-classicality can be systematically harnessed to enhance thermodynamic tasks. In light of this, we prove that any form of non-classicality can serve as a thermodynamic resource. In particular, any system that possesses quantum magic, coherence, or non-classical correlations can be leveraged to extract higher amounts of work than if the system does not possess such resources. The quantum thermodynamic advantages—quantified by the ratio between work extractable from a resource state and work extractable in its absence—increase with the \emph{resource robustness}.
We show that for any convex quantum resource theory, any resourceful state can yield a work-extraction advantage over all free states via a cyclic quench/thermalization protocol whose Hamiltonian is engineered from an optimal robustness witness. We illustrate concrete examples in which the robustness measures increase with the system's dimension, yielding quantum thermodynamic advantages that scale with it. In contrast, we also show that preparing a resource state (e.g., one with magic, coherence, or non-classical correlations) can be significantly more thermodynamically costly than preparing any state without such a resource. Concretely, there always exists a protocol that can prepare any non-resourceful state at significantly less work than it takes to prepare a resourceful state. Overall, our results provide operational meaning to robustness measures of quantum resources in terms of their thermodynamic costs and advantages.
\end{abstract}
\maketitle

There is a long and fruitful tradition in quantum information theory of identifying the traits responsible for advantages in solving a task. 
In quantum computing, for instance, entanglement~\cite{JozsaLinden2003,
Vidal} and quantum magic (nonstabilizerness)~\cite{gottesman1998heisenberg, gottesman1997stabilizercodesquantumerror,
BravyiKitaev2005, PhysRevLett.116.250501} are necessary traits for potential quantum advantages as, otherwise, dynamics can be efficiently simulated classically. 
In applications to various fields, entanglement can be more directly identified as the resource for quantum advantages~\cite{Ekert1991, Bennett1993, Bollinger1996,GiovannettiPRL2006, PhysRevLett.102.250501,  zhao2025entanglement}, and coherence can serve as catalyst in state preparation~\cite{PhysRevLett.132.180202, PhysRevLett.132.200201}. 

In quantum thermodynamics, coherence~\cite{lostaglio2015description, Korzekwa_2016, AlmanzaMarrero2025certifyingquantum}, non-classical correlations~\cite{PhysRevE.87.042123, PhysRevLett.89.180402, PhysRevA.88.052319, francica2017daemonic,QThermoAdvPRL2025}, many-body effects~\cite{Jaramillo2016, PhysRevLett.124.210603}, and contextuality~\cite{PhysRevLett.125.230603} can be leveraged to enhance the work output of a process. The rate at which quantum batteries can be charged is also enhanced by entanglement between battery cells~\cite{Binder2015Quantacell, PhysRevLett.118.150601}. 
For a recent and more thorough review of examples where quantum traits enhance thermodynamic tasks, see Ref.~\cite{Campbell_2026}.

Despite such results, it is typically hard to precisely and quantifiably pinpoint an advantage to a particular quantum trait---often, many quantum properties coexist in different degrees.
Quantum resource theories aim to formalize such studies by quantifying a system's various resources~\cite{RobustnessReview, Lostaglio2019, gour2024resourcesquantumworld}. In such theories, information-theoretic quantities are introduced to measure resources. 
A popular family of measures, known as resource robustness measures, quantify the resilience of a resource when a state mixes with other states [see the definition in Eq.~\eqref{eq:Robustness} below]. Robustness measures find operational meanings in state and channel discrimination tasks~\cite{TakagiPRL,Oszmaniec2019operational,AsymmetryPiani, RobustnessGPT, TakagiPRX, PhysRevA.109.042403, Resorcetheory_Informativeness,BiswasWinterRobustness_Distinguishability,Cavalcanti_Skrzpczyk_Discrimination}, distillation of resourceful states~\cite{TakagiGoldenPRL2019, Zhou2020GeneralStateTransitions,Kondra2025stochastic}, and classical simulation costs of a quantum process~\cite{Howard, Heinrich_2019, SheldonPRXQ2021}.

In this article, we provide an operational interpretation of resource robustness measures in thermodynamic contexts.
We do so by showing that any quantum resource state can provide a quantum thermodynamic advantage in work extraction that is quantified by the resource robustness. %
We show there always exists a protocol that can exploit such resources to extract more work than is possible with access to non-resourceful (e.g., classical) states~\footnote{We use ‘classical’ and ‘non-classical’ as shorthand for ‘free’ and ‘resourceful’ relative to the resource theory under consideration. This is a resource-theory-dependent convention (not a unique notion of classicality). In familiar cases (e.g., coherence in a fixed basis, stabilizer/magic settings), free states admit a standard classical description or efficient classical simulation.}.

Conversely, we also show that preparing quantum resource states can be significantly costlier than preparing non-resourceful (free) states. Specifically, there always exist protocols that prepare non-resourceful states at lower costs than those for resourceful states. In this way, we relate degrees of non-classicality of quantum systems, their potential as thermodynamic resources, and the thermodynamic costs to prepare them.

Particularly relevant prior articles include Refs.~\cite{PhysRevLett.134.050401, junior2025tradingathermalitynonstabiliserness, biswas2025all, hsieh2024generalquantumresourcesprovide, hsieh2025complete}.
In Ref.~\cite{PhysRevLett.134.050401}, it was shown that heat flow can certify quantumness in a system, relating thermodynamic processes to quantum resources.
Reference~\cite{junior2025tradingathermalitynonstabiliserness} characterizes regimes in which coupling to a heat bath can generate magic—quantified by robustness measures—from stabilizer states.
In Ref.~\cite{biswas2025all}, we show that access to quantumly correlated (steerable) states can be used to cool a system more than with access to classically correlated states, with an advantage quantified by the steerability robustness.
Refs.~\cite{hsieh2024generalquantumresourcesprovide, hsieh2025complete} introduced energy extraction tasks that can certify general quantum resources (including for channels) and characterize processes of state conversion. They further showed how such tasks can be used to quantify resources via a maximized relative energy advantage over free states. 
Our results complement the constructions in Refs.~\cite{hsieh2024generalquantumresourcesprovide, hsieh2025complete} by providing closed-form robustness bounds for a standard free-energy work functional, and by relating robustness to preparation/implementation costs.
Then, the innovations of the current article are
\begin{itemize}[leftmargin=1.6em]
\vspace{-3pt}\item[(i)] relating and bounding the relative thermodynamic advantage in terms of the resource robustness via an explicit witness–Hamiltonian protocol (allowing exponential scaling with system size for families with exponentially growing robustness),
\vspace{-3pt}\item[(ii)] proving the near-total depletion of resources post work extraction from pure states, and 
\vspace{-3pt}\item[(iii)] establishing bounds on the work costs of generating resourceful states and resource-generating operations.
\end{itemize}

To quantify the quantum resources in a state, we consider \emph{robustness measures}~\cite{RobustnessEntanglement, Steiner2003, Cavalcanti_incompatibility, Napoli2016, RobustnessReview, TakagiPRL, RobustnessGPT}. Given a convex set $\Free$ of free, resource-less states, the (generalized) robustness is defined by
\begin{align}
\label{eq:Robustness}
\Rob( \rho )  &\coloneqq \min_
{0 \leq \gamma; \  \tr{\gamma} = 1} 
\left\{ s \geq 0  \;\Big\Vert \;   \rho_s \coloneqq \frac{\rho + s \gamma }{1+s}  \in \Free   \right\}.
\end{align}
It follows that $\Rob( \rho )=0$ if and only if $  \rho \in \Free$, and it is strictly greater than $0$ otherwise. 
That is, states contained in the free set $\Free$ have zero robustness, while states outside of $\Free$ have $\Rob( \rho ) > 0$. Robustness quantifies how much mixing a resource state can withstand while still being resourceful. Sufficient mixing eventually removes all resources from a state.

The robustness can be conveniently re-expressed in terms of the following dual semidefinite program~\cite{Brandao2005, TakagiPRL}
\begin{align}
\label{eq:SDP}
\Rob (\rho)
=
\max_{Y \geq 0} 
\left\{ \tr{Y \rho} - 1
\;\Big\Vert \;
\tr{Y \sigma} \le 1 \quad \forall\, \sigma \in \Free
\right\},
\end{align}
where $Y$ is Hermitian.
If the \emph{witness} $Y^*$ is a solution to the previous equation, then $\Rob(\rho) = \tr{Y^* \rho} -1$. Note that $Y^*$ depends on the state and the resource of interest (e.g., the optimal witnesses for coherence or entanglement typically differ). We consider the generalized robustness throughout this work, not to be confused with the standard robustness or other robustness metrics.

The definition of the free set $\Free$ characterizes the quantum resource of interest. The robustness can quantify different quantum resources by different choices of $\Free$. For instance, if $\Free_\Sep$ defines the set of separable states, the corresponding $\mathcal{R}_{\Free_\Sep}(\rho)$ quantifies a state's entanglement~\cite{RobustnessEntanglement, Steiner2003}. If $\Free_\Incoh$ includes all states without coherence in a given basis, then  $\mathcal{R}_{\Free_\Incoh}(\rho)$ quantifies coherence in said basis~\cite{Napoli2016, AsymmetryPiani, PhysRevA.98.022328}. Other choices of sets of free states and operations define resource theories for magic~\cite{Howard, Heinrich_2019, Winter2022}, non-Gaussianity~\cite{RobustnessNonGaussianity2025}, measurements~\cite{Resorcetheory_Informativeness,Oszmaniec2019operational}, steerability~\cite{SteeringRPP, Cavalcanti_incompatibility} and incompatibility~\cite{Designolle_2019,Cavalcanti_Skrzpczyk_Discrimination}.

In the remainder of this article, we relate the resources in a quantum system, as quantified by robustness measures, to the performance of thermodynamic tasks.

\section{Thermodynamic advantages from quantum resources}
\label{sec:advantages}

Consider a system with Hamiltonian $H$ in a non-thermal state $\rho$. The maximum \emph{extractable work} from such a state given access to a bath at inverse temperature $\beta$ is characterized by
\begin{align}
\label{eq:WorkOut}
    W_H^{\out}(\rho) \coloneqq F_H(\rho) - F_H\big(\tau_H^\beta \big).
\end{align}
Here, 
\begin{align}
\label{eq:FreeEnergy}
F_H(\rho) \coloneqq \tr{H \rho} - \beta^{-1} S(\rho) 
\end{align}
is the system's free energy, $S(\rho)$ its von Neumann entropy, and $\tau_H^\beta \coloneqq e^{-\beta H}/\tr{e^{-\beta H}}$ is the thermal state. $\mathcal{N} W_H^{\out}(\rho)$ is the amount of work that can be extracted from asymptotically many copies $\rho^{\otimes \mathcal{N}}$ of the state~\cite{horodecki2013fundamental, workextraction1,workextraction2,workextraction3,workextraction4}. Any non-thermal state satisfies $ W_H^{\out}(\rho) \geq 0$, and can thus serve as a thermodynamical resource for work extraction.
We consider work extraction protocols under thermal operations, which assume free access to ideal thermal baths at inverse temperature $\beta$, access to unitary operations between systems and bath that globally conserve energy, and the possibility to discard systems~\cite{workextraction1}. 

To compare the potential to extract work from a resource state $\rho$ and resource-less states, we define 
\begin{align}
\label{eq:xiout}
    \xi_{\out}(\rho) \coloneqq \max_{H \geq 0} \frac{W_H^{\out }(\rho)}{\max_{\sigma \in \Free} W_H^{\out}(\sigma)}.
\end{align}
$\xi_{\out}(
\rho
)$ compares how much work can be extracted from a resource state $\rho$ relative to the maximum extractable work from free states, maximized over possible Hamiltonians. 
We take the maximization over all Hamiltonians bounded from below, but additional control constraints (e.g., bounded energy or locality) would restrict the optimization and can reduce the achievable advantage.

The first main result of this article relates $\xi_\out(\rho)$ to the state's resource robustness, for any resource.
\begin{theorem}[Thermodynamic advantages from non-classicality]
\label{TheoremAdvantage}
Consider the set $\Free$ of free resource-less states and a 
resource state $\rho$ with robustness $\Rob(\rho)$, as defined by Eq.~\eqref{eq:Robustness}. The work that can be extracted given access to a thermal bath at inverse temperature $\beta$ satisfies
\begin{align}
\label{eq:QThermoAdvantage}
    \xi_{\out}(\rho) \geq 1 + 
    \frac{\Rob(\rho)  - \lambda^{-1}\beta^{-1}   S(\rho) }{1 + \lambda^{-1}\beta^{-1} S\big(\tau_{\lambda Y^*}^\beta \big)}, 
\end{align}
provided that $\lambda \geq \beta^{-1} S(\rho) / \Rob(\rho)$.
Here, $Y^*$ is an optimal robustness witness, $\lambda$ is a multiplicative Hamiltonian parameter that influences energy scales in the work extraction protocol, and $\tau_{\lambda Y^*}^\beta$ is the corresponding thermal state. 

At zero temperature, 
\begin{align}
\label{eq:QThermoAdvantageIdeal}
 \xi_\out(\rho) \nobreak \geq \nobreak 1 \nobreak + \nobreak  \Rob(\rho). 
\end{align}. 
\end{theorem}

\begin{proof}
To prove Eq.~\eqref{eq:QThermoAdvantage}, we take a particular Hamiltonian, $H = \lambda Y^*$, where $\lambda > 0$ and $Y^* \geq 0$ is an optimal robustness witness that solves Eq~\eqref{eq:SDP}. Thus, $0 \leq \tr{H \sigma} = \lambda \tr{Y^* \sigma} \leq \lambda$ for any free state $\sigma \in \Free$, and $\tr{H \rho} =\lambda \tr{Y^* \rho} =\lambda \big(1+\Rob(\rho) \big)$ for a resourceful state. Let $\sigma_{\lambda Y^*} \in \Free$ be the free state that maximizes $W^\out_{H = \lambda Y^*}(\sigma)$.  Then, combining the previous two inline expressions with Eq.~\eqref{eq:xiout} gives that 
\begin{align}
\label{eq:Thm1Auxa}
    \xi_{\out}(\rho) 
  &\geq \frac{ \lambda \big( 1+\Rob(\rho) \big) - \beta^{-1} S(\rho) -  F_{\lambda Y^*}\big(\tau_{\lambda Y^*}^\beta \big) }{ \lambda -  \beta^{-1} S(\sigma_{\lambda Y^*}) - F_{\lambda Y^*}\big(\tau_{\lambda Y^*}^\beta \big) }  \nonumber \\
    &\geq \frac{ \lambda \big( 1+\Rob(\rho) \big) -  \beta^{-1} S(\rho) -  F_{\lambda Y^*}\big(\tau_{\lambda Y^*}^\beta \big) }{ \lambda - F_{\lambda Y^*}\big(\tau_{\lambda Y^*}^\beta \big) }.
\end{align}
The last inequality holds because the extractable work in Eq.~\eqref{eq:WorkOut} is non-negative, so the expression's denominator satisfies $0 \leq W^\out_{\lambda Y^*}(\sigma_{\lambda Y^*}) \leq \lambda - \beta^{-1} S(\sigma_{\lambda Y^*}) - F_{\lambda Y^*}\big(\tau_{\lambda Y^*}^\beta \big) \leq \lambda - F_{\lambda Y^*}\big(\tau_{\lambda Y^*}^\beta \big)$. 

Assume that $S(\rho) \leq \lambda \beta \Rob(\rho)$. Then, using that $(x-z)/(y-z) \geq x/y$ holds if $y-z \geq 0$, $z > 0$, and $x>y$, 
setting $x \equiv \lambda (1 + \Rob(\rho)) - \beta^{-1} S(\rho) + \beta^{-1} S(\tau_{\lambda^*}^\beta)$, $y \equiv \lambda + \beta^{-1} S(\tau_{\lambda Y^*}^\beta) $, and $z \equiv \lambda \tr{Y^* \tau_{\lambda Y^*}^\beta} \geq 0$, 
we obtain that 
\begin{align}
\label{eq:Thm1Auxb}
    \xi_{\out}(\rho) & \geq \frac{ \lambda \big( 1+\Rob(\rho) \big) - \beta^{-1} S(\rho) +  \beta^{-1} S\big(\tau_{\lambda Y^*}^\beta \big) }{ \lambda +  \beta^{-1} S\big(\tau_{\lambda Y^*}^\beta \big) } 
    \nonumber \\
    &= 1 + \frac{ \Rob(\rho)  - \lambda^{-1} \beta^{-1} S(\rho)}{ 1 + \lambda^{-1} \beta^{-1} S\big(\tau_{\lambda Y^*}^\beta \big) }.
\end{align}
This proves Eq.~\eqref{eq:QThermoAdvantage}. The entropic terms disappear at zero temperature, and $\xi_{\out}(\rho) \geq 1+ \Rob(\rho)$ holds then.

\end{proof}

Theorem~\ref{TheoremAdvantage} shows that a quantum thermodynamic advantage can be obtained from leveraging the non-classical resources in a state $\rho$. It also provides an operational meaning to the resource robustness in a thermodynamic context. The resource robustness quantifies the maximum advantage in work extraction from resourceful states and, as we describe in the next section, a protocol with a tailored engineered Hamiltonian can output work that grows with the state's robustness. 
Since the robustness can scale with the dimension of the state's Hilbert space, such an advantage can be exponential in system size. For a bounded Hamiltonian, $\|H \| \leq E$, the entropic corrections mean the advantage is capped; our bounds show how robustness competes with temperature and energy scales. 

\section{Work extraction from resource states}
\label{sec:examples}
\begin{figure*}[htbp]
    \centering
    \includegraphics[width=14cm]{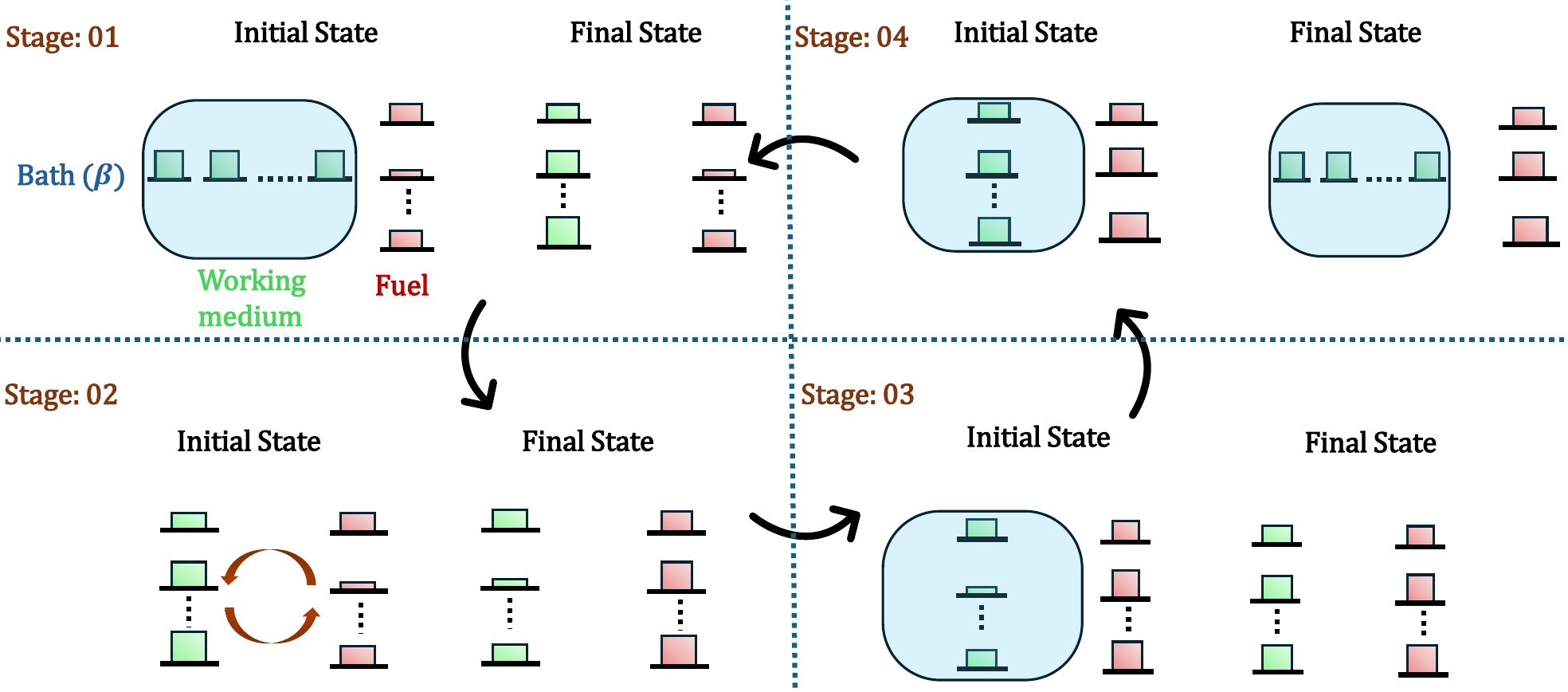}
    \caption{\label{fig:Work_extraction_protocol} 
This figure illustrates the work-extraction protocol described in Sec.~\ref{sec:examples}. In Stage~1 (Initial State), the working medium, characterized by a null Hamiltonian, is thermalized at the background bath at inverse temperature $\beta$, while the fuel system contains the resource state $\rho$. The Hamiltonian of the working medium is changed reversibly and isothermally, while maintaining contact with the bath, from $H=0$ to $H=\lambda Y^*$. 
In Stage~2, the system is isolated from the bath and a global swap operation is performed between the working medium and the fuel; this operation is energy-conserving and therefore incurs no work cost or gain. In Stage~3, contact with the bath is re-established and the working medium rethermalizes at inverse temperature $\beta$. Finally, in Stage~4, the Hamiltonian of the working medium is changed reversibly and isothermally, while remaining in contact with the bath, from $H=\lambda Y^*$ back to $H=0$, thereby restoring the initial configuration and enabling the protocol to be repeated with a fresh fuel state. The protocol outputs more work if the input state is nonclassical than if the input state is resourceless, with a relative advantage that increases with the resource robustness $\Rob(\rho)$.}
\end{figure*}

The proof of Theorem~\ref{TheoremAdvantage} suggests concrete work extraction protocols that yield thermodynamic advantages from resourceful states.
Consider a system that acts as a working medium and a fuel system. 
The fuel system initially stores the resource state $\rho$ and Hamiltonian $H_{\textnormal{fuel}} = \lambda Y^*$, where $Y^*$ is an optimal robustness witness for $\rho$. The working medium starts with the trivial Hamiltonian $H = 0$ and at a thermal state $\tau^\beta_0 = \id/d$.

The following cyclic protocol, illustrated in Fig.~\ref{fig:Work_extraction_protocol}, spends the fuel's resources to extract work: 
\begin{itemize}

    \item(a) Keeping the medium at thermal equilibrium with the bath, reversibly and isothermally change its Hamiltonian to $H = \lambda Y^*$. The net work change in this process is $\Delta W_{(a)} = F_0\big( \tau^\beta_0 \big) - F_{\lambda Y^*} \big( \tau^\beta_{\lambda Y^*} \big)$ (Stage 01 in Fig.~\ref{fig:Work_extraction_protocol}). 
    
    \item(b) Load the fuel's resource state $\rho$ onto the medium. Since fuel and working medium have identical Hamiltonians, this can be done by a globally energy preserving swap operation, so $\Delta W_{(b)} = 0$ (Stage 02 in Fig.~\ref{fig:Work_extraction_protocol}). 

    \item(c) Couple the working medium to the bath. If performed optimally, this process extracts $\Delta W_{(c)} = F_{\lambda Y^*}(\rho) - F_{\lambda Y^*}(\tau^\beta_{\lambda Y^*})$ work~\cite{workextraction1,workextraction2,workextraction3,workextraction4} (Stage 03 in Fig.~\ref{fig:Work_extraction_protocol}). 

    \item(d) Reversibly reset the working medium's Hamiltonian back to $H = 0$, at a cost $\Delta W_{(d)} = F_{\lambda Y^*}\big( \tau^\beta_{\lambda Y^*} \big) - F_{0} \big( \tau_{0}^\beta \big) \equiv - \Delta W_{(a)}$. This brings the medium back to its initial conditions, ready to accept another resource state for work extraction (Stage 04 in Fig.~\ref{fig:Work_extraction_protocol}). 
\end{itemize}

For any state $\rho$, such a sequence of quenches and thermalization yields an energy output 
\begin{align}
    W^\out(\rho) &\coloneqq \Delta W_{(a)} +  \Delta W_{(b)} +     \Delta W_{(c)} +     \Delta W_{(d)}  \\
    &= F_{\lambda Y^*}(\rho) - F_{\lambda Y^*}(\tau^\beta_{\lambda Y^*}) \nonumber \\
&=    \lambda \big( 1+\Rob(\rho) \big) - \beta^{-1} S(\rho) - F_{\lambda Y^*} \left(\tau^\beta_{\lambda Y^*}\right). \nonumber
\end{align}
Then, Theorem~\ref{TheoremAdvantage} implies that the relative work output between a resource state $\rho$ and a free state $\sigma \in \Free$ for such a protocol satisfies
\begin{align}
\label{eq:OutputWorkExtraction}
    \frac{W^\out(\rho)}{W^\out(\sigma)} &= \frac{\lambda \big( 1+\Rob(\rho) \big) - \beta^{-1}S(\rho) - F_{\lambda Y^*} \left(\tau^\beta_{\lambda Y^*}\right)}{ \lambda \big( 1+\Rob(\sigma) \big) - \beta^{-1}S(\sigma) - F_{\lambda Y^*} \left(\tau^\beta_{\lambda Y^*}\right)}  \nonumber \\
    &\gtrsim  1+\Rob(\rho),
\end{align}
for large enough $\lambda \beta$. Since the von Neumann entropy of a $d$-dimensional system satisfies $S \leq \log d$, it is enough to take $\lambda \beta \gg \log d$ to recover an advantage that scales with $\Rob(\rho)$.

\subsection*{Optimal witnesses}

Implementing the protocol described above requires knowledge of optimal robustness witnesses. Finding $Y^*$ for arbitrary states $\rho$ of high-dimensional systems is typically intractable in practice (it involves a semidefinite program, which is efficient to simulate, but intractable for large number of qubits, for instance). 
However, for pure states $\rho = \ket{\psi}\!\bra{\psi}$, an optimal witness can be expressed in terms of a rank-1 projector~\cite{Regula2018Convex, TakagiPRL, PhysRevA.98.022328},
\begin{align}
\label{eq:OptimalWitness}
Y^* = c \ket{y}\!\bra{y}, 
\end{align}
where $\ket{y}$ is a normalized state. 
The state's robustness is $\Rob(\rho) = c \, |\!\bra{y} \! \psi \rangle|^2 - 1$.

After the work extraction protocol with such an optimal witness, the system ends in the thermal state $\tau^\beta_{\lambda Y^*} = \tau^\beta_{\lambda c \ket{y}\!\bra{y}}$. 
Does this state contain any resources from which one could extract work (at the same fixed inverse temperature)? The following result answers this question in the negative.
\begin{theorem}
\label{TheoremThermal}
Consider a $d$-dimensional pure state $\rho = \ket{\psi}\!\bra{\psi}$ with robustness $\Rob(\rho)$ with respect to the free set $\Free$. Let $Y^* = c \ket{y}\!\bra{y}$ be an optimal witness for such a resource. After an optimal work extraction protocol where the system ends in a thermal state $\tau^\beta_{\lambda Y^*}$, it holds that
\begin{align}
    \mathcal{R}_{\Free'}\big( \tau^\beta_{\lambda Y^*} \big) \leq \frac{1}{d-1},
\end{align}
where $\mathcal{R}_{\Free'}$ is the robustness of \emph{any} resource theory for which the maximally mixed state is free, $\id/d \in \Free'$. 
\end{theorem}

\begin{proof}
    The thermal state for a rank-one Hamiltonian $H = \lambda Y^* = c \lambda \ket{y}\!\bra{y}$ is
\begin{align} 
\label{eq:AuxThermal}
 \tau^\beta_{\lambda Y^*}   &= \frac{ \id - \left(1-  e^{-\beta \lambda c} \right)  \ket{y}\!\bra{y}  }{d -\left(1-  e^{-\beta \lambda c} \right) }.
\end{align}
Consider a resource theory defined by the free set $\Free'$ such that $\id/d \in \Free'$ (e.g., for coherence, entanglement, and magic). Let $Z$ be an optimal witness of state $\tau^\beta_{\lambda Y^*}$ for such a resource. Since the maximally mixed state is resourceless, then $\tr{ Z \id/d} \leq 1$. Recall, too, that $Z \geq 0$, so $- \tr{Z \ket{y}\!\bra{y}} \leq 0$.
Then, using Eq.~\eqref{eq:AuxThermal} we find that
\begin{align}
    \mathcal{R}_{\Free'} \left( \tau^\beta_{\lambda Y^*} \right) &= \tr{Z \tau^\beta_{\lambda Y^*}} - 1  \nonumber \\
    &=  \frac{ \tr{Z\id} - \left(1-  e^{-\beta \lambda c} \right)  \tr{ Z \ket{y}\!\bra{y}}  }{d -\left(1-  e^{-\beta \lambda c} \right) } - 1  \nonumber \\
    & \leq  \frac{d}{d -\left(1-  e^{-\beta \lambda c} \right) } - 1 \leq \frac{1}{d-1}.
\end{align}
\end{proof}

Theorem~\ref{TheoremThermal} shows that virtually no resources (for theories where the maximally mixed state is free) remain at the end of the work extraction protocol described above for high-dimensional systems. After extracting work from the robustness $\Rob(\rho)$ from a resource characterized by the free set $\Free$, little work can be extracted from any resource such that $\id/d\in \Free'$ for large $d$. Then, the optimal work-extraction approach is to identify the type of resource $\Free$ for which $\mathcal{R}_{\Free}(\rho)$ is maximum and apply the optimal protocol for such a resource theory.

This pinpoints the problem of extracting work from resourceful states to identifying optimal witnesses $Y^*$ that can be used to construct a Hamiltonian $H \propto Y^*$ and quantifying the corresponding resource robustness.

\vspace{8pt}
\noindent  {\emph{Example: work extraction from coherent states}.---}The robustness of coherence is characterized by the free set $\Free_\Incoh$ by states incoherent in a given basis. Coherent states then have non-zero $\mathcal{R}_{\Free_\Incoh}$~\cite{Napoli2016, AsymmetryPiani, PhysRevA.98.022328}.

Consider coherence in the computational basis and a pure state $\rho = \ket{\psi}\!\bra{\psi}$, where $\ket{\psi} = \sum_{j=1}^d c_j \ket{j}$. The robustness of coherence is $\mathcal{R}_{\Free_\Incoh}(\rho) = \big( \sum_{j=1}^d |c_j|\big)^2 -  1$~\cite{Napoli2016}. The optimal witness, a rank-one operator as in Eq.~\eqref{eq:OptimalWitness}, simplifies to $Y^* \propto \ket{\psi}\!\bra{\psi}$ for maximally-coherent states with $c_j = 1/\sqrt{d}$ (known as golden states~\cite{TakagiGoldenPRL2019}). Then, $\mathcal{R}_{\Free_\Incoh}(\rho) = d-1$. The protocol described at the beginning of this section would then discriminate the maximally coherent state from incoherent states by a relative work extraction ratio that grows linearly with the number of states in the superposition.

\vspace{8pt}
\noindent {\emph{Example: work extraction from magic states.---}}Let $\Free_\Stab$ denote the set of stabilizer states. An $N$-qubit state $\ket{\psi}$ is a stabilizer state if it is the joint eigenstate with eigenvalue $+1$ of $N$ independent and commuting Pauli operators $\{g_1,\dots,g_N\} \in \mathcal{P}_N$, where $\mathcal{P}_N$ denotes the Pauli group. That is, $g_k  \ket{\psi} = \ket{\psi}$ for $ k=1,\dots, N$. A mixed state $\rho$ is a stabilizer state if it can be expressed as a convex combination of pure stabilizer states. $\mathcal{R}_{\Free_\Stab}(\rho)$ quantifies the non-stabilizerness, or magic, in a state $\rho$~\cite{PhysRevLett.115.070501, Howard, Heinrich_2019, Winter2022}.

Stabilizer states can be efficiently prepared by Clifford circuits, and Clifford circuits can be efficiently simulated  classically~\cite{gottesman1998heisenberg, PhysRevA.70.052328}. Thus, quantum magic is a necessary ingredient for quantum computational advantages over classical methods.  
In fact, the generalized robustness of magic relates to the runtime of certain classical algorithms to simulate a quantum process---the higher the magic, the harder it is for such classical algorithms to simulate the quantum process~\cite{Howard, Heinrich_2019, SheldonPRXQ2021}.

Consider $N$ copies of a $\ket{T}$ state, $\ket{\psi} = \ket{T}^{\otimes N}$. $\ket{T}$ states are magic states obtained from applying the non-Clifford gate $T$ [$R_z(\pi/4)$] on a stabilizer state,
\begin{align}
    \ket{T} \coloneqq T \ket{+} = \left(\ket{0} + e^{i\pi/4}\ket{1} \right)/\sqrt{2}.
\end{align}
In turn, $\ket{T}$ states can be used to implement $T$ gates by a Clifford operation. Adding $T$ gates to the Clifford group makes a circuit universal, i.e., capable of approximating any unitary operation, and no longer classically simulable. $T$ gates are harder to implement than Clifford operations, thus the value of $\ket{T}$ states in quantum computing~\cite{BravyiKitaev2005, PhysRevX.2.041021}.

Reference~\cite{SheldonPRXQ2021} shows that an optimal witness of magic for any single-qubit magic state takes the form
\begin{align}
Y_1^* = \frac{\id + (X+Y)/\sqrt{2}}{1+ 1/\sqrt{2}},
\end{align}
where 
$\{X,Y,Z\}$ are the Pauli matrices. A simple optimization over $q$ yields the optimal witness for $\ket{T}$ states. 
Ref.~\cite{SheldonPRXQ2021} also shows that $\big( \mathcal{R}_{\Free_\Stab}+1 \big)$ is multiplicative over single-qubit states, which implies that an optimal witness for $\ket{T}^{\otimes N}$ is a tensor product, $Y_N^* = \bigotimes_{j=1}^N Y_j^*$ where $Y_j^*=Y_1^*$ for all $j$. 
Then, the robustness of $N$ copies of a $T$ state is
$\mathcal{R}_{\Free_\Stab}\big(\ket{T}\!\bra{T}^{\otimes N}\big) = \left( 4 - 2\sqrt{2} \right)^N -1 \approx 1.1716^N-1$.

Using $\ket{T}\!\bra{T}^{\otimes N}$ as inputs to the work extraction protocol described in the previous section would yield a separation $\xi_{\out}\big(\ket{T}\!\bra{T}^{\otimes N}\big)\sim 1.1716^N$ with the extractable work from stabilizer states. 
However, this may not be the optimal approach to extract work! 

The generalized robustness of coherence in the computational basis of $\ket{T}\!\bra{T}^{\otimes N}$ is $\mathcal{R}_{\Free_\Incoh} = 2^N - 1$, which can be seen from the expression for the generalized robustness of coherence in the computational basis described above. Then, extracting work from $\ket{T}\!\bra{T}^{\otimes N}$ by using coherence as the resource, with the corresponding optimal witness for the robustness of coherence, yields a separation $\xi_{\out}\big(\ket{T}\!\bra{T}^{\otimes N}\big)\sim 2^N - 1$, higher than the one obtained from the optimal magic witness. 

The best work-extraction protocol requires identifying the type of resource that the state $\rho$ possesses the most. Once the optimal resource is known, the optimal protocol involve strokes of thermalization and quenches to a Hamiltonian proportional to the optimal resource witness.

In the previous examples, we considered optimal resource witnesses. However, even sub-optimal resource witnesses $Y$ such that $0 < \tr{Y \rho} - 1 \leq \tr{Y^* \rho} -  = \Rob(\rho)$ yield an advantage in work extraction from resourceful states. This can be seen, for instance, by using the sub-optimal Hamiltonian $H' = \lambda Y$ in the protocol described at the beginning of this section. 

\section{Thermodynamic costs of quantum resources}
\label{sec:costs}

In the previous two sections, we showed that quantum resource states can yield advantages versus resourceless states in work extraction tasks. However, there is a natural counter-balance of such results in terms of the thermodynamic resources necessary to prepare resources.

The minimum \emph{work cost} required to prepare a state $\rho$ from a thermal state is given by
\begin{align}
    W_H^{\cost}(\rho) =  F_H(\rho) - F_H\big(\tau_H^\beta \big),
\end{align}
positive whenever $\rho$ is not the thermal state $\tau_H^\beta$~\cite{workextraction1,workextraction2,workextraction3,workextraction4}.  

To compare the energetic resources needed to prepare free and quantum states from thermal states, let us consider a quantity analogous to Eq.~\eqref{eq:xiout}:
\begin{align}
    \xi_{\cost}(\rho) \coloneqq \min_{H \geq 0 } \frac{\max_{\sigma \in \Free} W_H^{\cost}(\sigma)}{W_H^{\cost}(\rho)}.
\end{align}
$\xi_{cost}$ compares the minimum relative work cost to prepare the most costly free state ($\max_{\sigma \in \Free} W_H^{cost}(\sigma)$) to the cost to prepare a quantum state ($W_H^{cost}(\rho)$), minimized over all Hamiltonians. If $\xi_{cost}(\rho)$ is small, it means that there exists a resourceless state that is much less costly to prepare than the quantum resourceful state, under some protocol. 

Then, the next result follows from Theorem~\ref{TheoremAdvantage}.
\begin{corollary}[Thermodynamic costs of non-classicality]
\label{CorollaryCost}
Consider the set $\Free$ of free resource-less states and a resource state $\rho$ with robustness $\Rob(\rho)$. 
The work cost required to prepare $\rho$ given access to a thermal bath at inverse temperature $\beta$ satisfies
\begin{align}
\label{eq:QThermoCost}
    \xi_{\cost}(\rho) \leq  \frac{ 1 +  \lambda^{-1}\beta^{-1} S(\tau_{\lambda Y^*}^\beta) }{ 1 + \Rob(\rho) + \lambda^{-1}\beta^{-1} \left( S(\tau_{\lambda Y^*}^\beta) - S(\rho) \right)  },
\end{align}
with $\lambda \geq \beta^{-1} S(\rho) / \Rob(\rho)$.
Here, $Y^*$ is an optimal robustness witness, $\lambda$ is a multiplicative Hamiltonian parameter that influences energy scales in the work extraction protocol, and $\tau_{\lambda Y^*}^\beta$ is the corresponding thermal state.

At zero temperature,
\begin{align}
\label{eq:QThermoCostIdeal}
    \xi_\cost(\rho) \nobreak \leq \nobreak \frac{1}{1+\Rob(\rho)}.
\end{align}

\end{corollary}

Corollary~\ref{CorollaryCost} implies that, for any resource state with robustness $\Rob(\rho)$, there exists protocols that can prepare any free state $\sigma \in \Free$ at lower work cost. For very resourceful pure states $\rho$, where the robustness can increase with the dimension of the Hilbert space, the extra cost to prepare such a state is exponentially large in system size.

Similar results constrain the thermodynamic costs of implementing non-thermal quantum operations that output resourceful states.
Consider, for concreteness, a resourceless and pure input state $\sigma_{\text{in}} \in \Free$. $\sigma_{\text{in}}$ could be, for instance, the input state to a quantum circuit where $N$ qubits are initialized in the $\ket{0}$ state in the computational basis, $\sigma_{\text{in}} = (\ket{0}\!\bra{0})^{\otimes N}$. Such a state has no magic, entanglement, or coherence in the computational basis. 

Let $\mathcal{E}$ be a resource quantum channel capable of generating non-classicality, and let $\Omega \in \mathcal{M}_\Free$ denote a channel that takes free states into free states, belonging to the set $\mathcal{M}_\Free$ of resource non-increasing maps, with respect to the resource theory defined by the set $\Free$. 
Then, the parameter
\begin{align}
\eta_{\textnormal{cost}}\big( \mathcal{E} (\sigma_{\text{in}}) \big) \coloneqq 
    \min_{H \geq 0}  \, \max_{\Omega \in \mathcal{M}_\Free } \frac{F_H \big( \Omega(\sigma_{\text{in}}) \big) - F_H (\sigma_{\text{in}})}{F_H \big( \mathcal{E}(\sigma_{\text{in}}) \big) - F_H (\sigma_{\text{in}})} 
\end{align}
compares the thermodynamic cost to implement the most costly free channel $\Omega$ to the cost to implement the resource-generating channel $\mathcal{E}$, minimizing over all protocols with fixed Hamiltonian $H$.

\begin{theorem}[Thermodynamic costs of non-classical state generation]
\label{TheoremCostChannel} 
Consider a resourceless and pure input state $\sigma_{\text{in}} \in \Free$. Let $\Omega \in \mathcal{M}_\Free$ be a free operation, and let $\mathcal{E}$ be quantum channel that can generate resources. 
The work cost of implementing any free operation versus the work cost of outputting the resourceful state $\mathcal{E} \big( \sigma_{\text{in}} \big)$ given access to a thermal bath at inverse temperature $\beta$ satisfies
\begin{align}
   \eta_{\textnormal{cost}}\big( \mathcal{E} (\sigma_{\text{in}}) \big) & \leq  \frac{1}{  \Rob \big( \mathcal{E}(\sigma_{\text{in}}) \big) - \lambda^{-1} \beta^{-1} S\big( \mathcal{E}(\sigma_{\text{in}}) \big)   }, 
\end{align}
provided  that $\lambda \geq \beta^{-1} S\big( \mathcal{E}(\sigma_{\text{in}}) \big) / \Rob \big( \mathcal{E}(\sigma_{\text{in}}) \big)$. 
Here, $\lambda$ is a multiplicative Hamiltonian parameter, $H = \lambda Y^*$, where $Y^*$ is an optimal robustness witness. 

For zero-entropy output states $\mathcal{E}(\sigma_{\text{in}})$ or at zero temperature, 
\begin{align}
   \eta_{\textnormal{cost}}\big( \mathcal{E} (\sigma_{\text{in}}) \big) \leq \frac{1}{\Rob \big( \mathcal{E}(\sigma_{\text{in}}) \big) }.
\end{align}
\end{theorem}

\begin{proof}
Let $Y^*$ be an optimal resource witness for the output state $\mathcal{E}(\sigma_{\text{in}})$, so that $\Rob\big( \mathcal{E}(\sigma_{\text{in}}) \big) =  \tr{Y^* \mathcal{E}(\sigma_{\text{in}})} - 1$. Since $Y^* \geq 0$,  and $\sigma_{\text{in}}$ and $\Omega \big( \sigma_{\text{in}} \big)$ are free states, it holds that $0 \leq \tr{\sigma_{\text{in}} Y^*} \leq 1$ and $0 \leq \tr{ \Omega \big( \sigma_{\text{in}} \big) Y^*} \leq 1$. Then, 
\begin{align}
&   \eta_{\textnormal{cost}}\big( \mathcal{E} (\sigma_{\text{in}}) \big)  \leq 
   \max_{\Omega \in \mathcal{M}_\Free } \frac{F_{\lambda Y^*} \big( \Omega(\sigma_{\text{in}}) \big) - F_{\lambda Y^*} (\sigma_{\text{in}})}{F_{\lambda Y^*} \big( \mathcal{E}(\sigma_{\text{in}}) \big) - F_{\lambda Y^*} (\sigma_{\text{in}})}   \nonumber \\
    & \qquad \leq   \max_{\Omega \in \mathcal{M}_\Free } \frac{\lambda - \beta^{-1} S\big( \Omega(\sigma_{\text{in}}) \big) }{ \lambda \Big( 1 + \Rob \big( \mathcal{E}(\sigma_{\text{in}}) \big) \Big)  - \beta^{-1} S\big( \mathcal{E}(\sigma_{\text{in}}) \big)  -  \lambda }  \nonumber \\
   & \qquad =   \frac{1}{  \Rob \big( \mathcal{E}(\sigma_{\text{in}}) \big) - \lambda^{-1} \beta^{-1} S\big( \mathcal{E}(\sigma_{\text{in}}) \big)   }. 
\end{align}
\end{proof}

For simplicity, consider pure output states or zero temperature, in which case $\eta_{\textnormal{cost}}\big( \mathcal{E} (\sigma_{\text{in}}) \big) \leq 1/ \Rob \big( \mathcal{E}(\sigma_{\text{in}}) \big)$.
Theorem~\ref{TheoremCostChannel} then shows that there exists some characteristic system Hamiltonian with respect to which the work cost to generate any resourceless state is cheaper than implementing the resourceful operation $\mathcal{E}$ on $\sigma_{\text{in}}$.

The previous result quantifies the costs of implementing $\mathcal{E}$ when acting on a particular input state. But, what about the thermodynamic costs of implementing the resourceful channel itself? How does it relate to the channel's resourcefulness?
We address this question next.

\subsection*{Thermodynamic costs of resourceful operations}

One way to implement $\mathcal{E}$ is as follows. 
Consider a bipartite system $AB$ with a tensor product structure on a Hilbert space $\mathcal{H}_A \otimes \mathcal{H}_B$, where $d = \dim(\mathcal{H}_A) = \dim(\mathcal{H}_B)$. Define the bipartite ($\mathcal{E}$-dependent) state 
\begin{align}
    J_{AB}^{\mathcal{E}} \coloneqq \left( \mathcal{E} \otimes \id_B  \right) \ket{\Phi}\!\bra{\Phi},
\end{align}
where $\ket{\Phi} \coloneqq \frac{1}{\sqrt{d}}\sum_{j=1}^{d} \ket{j}\!\otimes\!\ket{j}$. $J_{AB}^{\mathcal{E}}$ is the Choi state~\cite{watrous2018theory, frembs2024variations}, which characterizes the channel's action by $\mathcal{E}(\rho) = d \, \text{Tr}_B \left( \left(\id_A \otimes \rho^T \right) J_{AB}^\mathcal{E} \right)$, where $\rho^T$ denotes the transpose. Access to copies of $J_{AB}^\mathcal{E}$  allow implementing $\mathcal{E}$.

We use the work cost of preparing the Choi state $J_{AB}^\mathcal{E}$ from a thermal state,
\begin{align}
\label{eq:WorkProxyChannel}
    W_{H_{AB}}^{\cost}(\mathcal{E}) \coloneqq  F_{H_{AB}}\left( J_{AB}^\mathcal{E} \right) - F_{H_{AB}} \left( \tau_{H_{AB}}^\beta \right),
\end{align}
as a proxy for the implementation costs of $\mathcal{E}$, where $H_{AB}$ is a Hamiltonian on $AB$. $W_{H_{AB}}^{\cost}(\mathcal{E})$ quantifies the cost of implementing the channel via its Choi representation. (See also Ref.~\cite{faist2021thermodynamic}.)
Given a resourceful channel $\mathcal{E}$, consider
\begin{align}
\eta_{\textnormal{cost}}\big( \mathcal{E} \big) \coloneqq 
    \min_{H_{AB} \geq 0}  \,  \frac{ \max_{\Omega \in \mathcal{M}_\Free}  W_{H_{AB}}^{\cost} ( \Omega ) }{ W_{H_{AB}}^{\cost}(\mathcal{E})  },  
\end{align}
which compares the work required to implement (via Choi states) the most costly resourceless channel $\Omega \in \mathcal{M}_\Free$ to the cost of implementing the resourceful channel $\mathcal{E}$. $\eta_{\textnormal{cost}}\big( \mathcal{E} \big)$ provides a lower bound on the relative cost of the channels (i.e., other implementations could yield an even lower minimum relative cost).

As a last ingredient, the (generalized) robustness of a channel $\mathcal{E}$ is defined similarly to states~\cite{TakagiPRX, GourWinter2019, liu2019resourcetheoriesquantumchannels}: 
\begin{align}
\label{eq:RobustnessChannel}
\mathcal{R}_{\mathcal{M}_\Free}( \mathcal{E} )  &\coloneqq \min_{\kappa} 
\left\{ s \geq 0 \,\, \textnormal{s.t.} \, \,   \mathcal{E} \coloneqq \frac{\mathcal{E} + s \kappa }{1+s}  \in \mathcal{M}_\Free   \right\}.
\end{align}
Here, $\mathcal{M}_\Free$ defines a set of free channels and the minimization is over any possible channel $\kappa$. 
The dual formulation of the channel robustness is conveniently expressed in terms of the channel's Choi matrix as $\mathcal{R}_{\mathcal{M}_\Free} \left( \mathcal{E} \right) = \max_{ Z \geq 0} \left\{  \tr{ Z  J_{AB}^{\mathcal{E}} } - 1 \right\}$, where the witnesses $Z$ acting on $\mathcal{H}_A \otimes \mathcal{H}_B$ satisfy $\tr{ Z J_{AB}^{\Omega} } \leq 1$ for free $\Omega \in \mathcal{M}_\Free$~\cite{TakagiPRX}. 


Then, our final result bounds the implementation costs of resourceful channels in terms of their robustness:  
\begin{theorem}[Thermodynamic costs of non-classical channels]
\label{TheoremCostChannelGeneral}
Let $\mathcal{E}$ be resourceful quantum channel. The work cost of implementing any free $\Omega \in \mathcal{M}_\Free$ versus the work cost of implementing $\mathcal{E}$ given access to a thermal bath at inverse temperature $\beta$ satisfies
\begin{align}
\label{eq:QThermoCosChannel}
\eta_{\textnormal{cost}} \big( \mathcal{E} \big) & \leq   \frac{ 1 + \lambda^{-1} \beta^{-1} S\big( \tau_{\lambda Z^*}^\beta \big)}{  1 + \mathcal{R}_{\mathcal{M}_\Free}\big( \mathcal{E} \big) + \lambda^{-1} \beta^{-1} \left(  S\big( \tau_{\lambda Z^*}^\beta \big)  - S( J_{AB}^\mathcal{E} ) \right) },
\end{align}
provided that $\lambda \geq \beta^{-1} S( J_{AB}^\mathcal{E} ) / \mathcal{R}_{\mathcal{M}_\Free}\big( \mathcal{E} \big)$~\footnote{The entropy is bounded by the Choi rank, $S\big( J_{AB}^\mathcal{E} \big) \leq \log \big( \textnormal{rank}  \big( J_{AB}^\mathcal{E} \big) \big)$, which often provides insight into properties of a process~\cite{ChoiRank2024}. This implies that the entropy is zero for unitary transformations.}.
Here, $Z^*$ is an optimal robustness witness for the channel, $\lambda$ is a multiplicative Hamiltonian parameter that influences energy scales in the preparation protocol, and $\tau_{\lambda Z^*}^\beta$ is the corresponding thermal state. 
At zero temperature,
\begin{align}
\label{eq:QThermoCostChannelIdeal}
\eta_{\textnormal{cost}} \big( \mathcal{E} \big) & \leq   \frac{ 1  }{  1 + \mathcal{R}_{\mathcal{M}_\Free}\big( \mathcal{E} \big)  }.
\end{align}

\end{theorem}

\begin{proof}
We choose a particular Hamiltonian, $H_{AB} = \lambda Z^*$, where $Z^*$ is an optimal channel robustness witness. Then, using Eq.~\eqref{eq:WorkProxyChannel}, we get that
\begin{align}
\label{eq:aux-TheoremCostChannelGeneral}
\eta_{\textnormal{cost}}\big( \mathcal{E} \big) &\leq \frac{ \max_{\Omega \in \mathcal{M}_
    \Free}  W_{\lambda Z^*}^{\cost} ( \Omega ) }{ W_{\lambda Z^*}^{\cost}(\mathcal{E})  } \nonumber \\
    &= \max_{\Omega \in \mathcal{M}_
    \Free }    \frac{ F_{\lambda Z^*} \left( J_{AB}^\Omega \right) - F_{\lambda Z^*} \left( \tau_{\lambda Z^*}^\beta \right) }{ F_{\lambda Z^*} \left( J_{AB}^\mathcal{E} \right) - F_{\lambda Z^*} \left( \tau_{\lambda Z^*}^\beta \right) }.  
\end{align}
Equation~\eqref{eq:aux-TheoremCostChannelGeneral} is expressed in terms of states, with expressions similar to those in the proof of Theorem~\ref{TheoremAdvantage}. 
All resourceless channels $\Omega \in \mathcal{M}_\Free$ have Choi states that satisfy $\tr{Z^* J_{AB}^\Omega} \leq 1$, while $\tr{Z^* J_{AB}^\mathcal{E}} = \mathcal{R}_{\mathcal{M}_\Free}\big( \mathcal{E} \big) + 1$ for the resourceful state [see the dual formulation of the channel robustness after Eq.~\eqref{eq:RobustnessChannel}]. Following a similar derivation to that of Theorem~\ref{TheoremAdvantage} then shows Eq.~\eqref{eq:QThermoCosChannel}. 
\end{proof}
Theorem~\ref{TheoremCostChannelGeneral} bounds the thermodynamic costs of implementing a resourceful channel in terms of the channel robustness. More resourceful channels have a higher relative cost. 
(An immediate corollary of Theorem~\ref{TheoremCostChannelGeneral} is that resourceful channels can be leveraged for advantages in work-extraction, analogously to Theorem~\ref{TheoremAdvantage}.)

For instance, consider magic as the quantum resource, i.e., let $\Free_\Stab$ denote the set of stabilizer states. Circuits acting on $N$-qubits that implement channels $\Omega \in \mathcal{M}_{\Free_\Stab}$ are efficiently simulable by classical means. 
Quantum advantage thus requires circuits that implement resourceful channels $\mathcal{E}$.
Theorems~\ref{TheoremCostChannel} and~\ref{TheoremCostChannelGeneral} link the thermodynamic costs of implementing such resourceful channels to their resource content. The best case for quantum computing---one where the robustness is as high as possible, exponential in the number of qubits--- is the worst case in a thermodynamic sense: one where free states can be exponentially less energetically costly to prepare than the very magic ones.

\section{Discussion}
\label{sec:discussion}

Our results provide operational meanings to generalized resource robustness measures in a thermodynamic context.
Theorem~\ref{TheoremAdvantage} presents quantum resourceful states defined by the free set $\Free$ (e.g., states with magic, entanglement, or coherence, among others) as a thermodynamic resource. 
Our results do not mean that any protocol extracts more work from resourceful states. Instead, they inform protocols that, via quenches to a witness Hamiltonian $H \propto Y^*$, extract more work from a resourceful state $\rho$ than from any resourceless state $\sigma \in \Free$. The suggested protocols are tailored to the corresponding resource. The ratio of the extractable works is then characterized by the resource robustness $\Rob(\rho)$. 
Operationally, 
$\Rob(\rho)$ lower-bounds the optimized work-advantage ratio $\xi_\out(\rho)$ achievable by choosing 
$H = \lambda Y^*$ and running the quench/thermalization cycle.
Theorem~\ref{TheoremThermal} further shows that resources are exhausted when performing the optimal work extraction protocol from a pure resourceful state.

We illustrated the results for resourceful coherent or magic states, where optimal witnesses are known and can be expressed in terms of simple rank-1 projectors. Even though T states are a paradigmatic example of magic states, they have more computational-basis coherence than magic. As a result, the best work-extraction protocol is tailored to their coherence rather than their magic.
The only restriction we placed on the allowable Hamiltonians is that they are bounded from below. A cap on the maximum energy can limit the thermodynamic advantage for mixed input states, as can be seen from Theorem~\ref{TheoremAdvantage}. Placing extra constraints on the Hamiltonians, such as locality, would further bound the thermodynamic advantage, as the optimal witness may not be implementable then.

Corollary~\ref{CorollaryCost} and Theorems~\ref{TheoremCostChannel} and~\ref{TheoremCostChannelGeneral} counterbalance the aforementioned results, by comparing the thermodynamic costs of preparing resourceless and resourceful states. There always exists a protocol, defined by quenches with an optimal witness Hamiltonian, that prepares resourceless states at less cost than resourceful ones, with a ratio characterized by the robustness of the quantum resource state. More speculatively, Theorems~\ref{TheoremCostChannel} and~\ref{TheoremCostChannelGeneral} could have implications for quantum algorithms with the potential for quantum advantage, which require generating magic and entanglement. The theorems imply that there exist protocols that implement non-magic or non-entangling operations at a cheaper thermodynamic cost, potentially motivating the study of classically simulable (magic- and/or entanglement- free) circuits that incur lower thermodynamic costs than the quantum, magic/entangling generating ones under some Hamiltonian choices.


\section*{Acknowledgments}
This work was supported by the U.S. Department of Energy (DOE), Office of Science, Basic Energy Sciences program (award No. DE-SCL0000157) 
and the 
Advanced Scientific Computing Research program under project TEQA. 
LPGP also acknowledges support from the U.S. DOE, Office of Science, National Quantum Information Science Research Centers, Quantum Science Center.

\bibliography{main}

\end{document}